\documentclass[11pt]{article}
\usepackage{mathpazo}
\usepackage[dvips, paper=letterpaper, top=1in, bottom=1in, left=1in, right=1in]{geometry}

\usepackage[utf8]{inputenc} 
\usepackage[T1]{fontenc}    
\usepackage{url}            
\usepackage{booktabs}       
\usepackage{amsfonts}       
\usepackage{nicefrac}       
\usepackage{microtype}      
\usepackage{xcolor}         
\usepackage{bbm}
\usepackage{multirow}
\usepackage{amsmath,amssymb,amsthm,amsfonts,latexsym,bbm,xspace,graphicx,float,mathtools,dsfont}
\usepackage[backref,colorlinks,citecolor=blue,bookmarks=true]{hyperref}
\usepackage{algorithm}
\usepackage{algorithmic}

\usepackage{appendix}
\usepackage{cleveref}
\usepackage{graphicx}
\usepackage{accents}

\newtheorem{theorem}{Theorem}

\newtheorem{lemma}{Lemma}

\newtheorem{proposition}{Proposition}

\newtheorem{definition}{Definition}

\newtheorem{remark}{Remark}

\newenvironment{prevproof}[2]{\noindent {\em {Proof of {#1}~\ref{#2}:}}}{$\Box$\vskip \belowdisplayskip}

\newcommand{\yangnote}[1]{{\color{black}{#1}}}

\newcommand{\rev}{\textsc{Rev}}
\newcommand{\notshow}[1]{{}}

\newcommand{\latentDist}{\widehat{D}_z}
\newcommand{\lipshitz}{\mathcal{L}}
\newcommand{\mlDistI}{\widehat{D}_i}
\newcommand{\latentDistI}{\widehat{D}_{z,i}}
\newcommand{\givenMech}{\widehat{M}}
\newcommand{\query}{\mathcal{Q}}
\newcommand{\condsample}{\mathcal{S}}
\newcommand{\qanswer}{\hat{y}}
\newcommand{\tanswer}{y}
\newcommand{\normal}{\mathcal{N}}
\newcommand{\disteq}{\overset{d}{=}}
\newcommand{\E}{\mathbb{E}}
\newcommand{\norm}[1]{\left\lVert#1\right\rVert_2}
\newcommand{\normI}[1]{\left\lVert#1\right\rVert_\infty}
\newcommand{\lsv}[1]{\sigma_{max}(#1)}
\newcommand{\ssv}[1]{\sigma_{min}(#1)}
\newcommand{\Tr}[1]{\mathrm{Tr}(#1)}
\newcommand{\Influ}[1]{\textsc{Inf}(#1)}
\newcommand{\var}{\mathrm{Var}}

\title{Recommender Systems meet Mechanism Design}
\author{Yang Cai\footnote{Supported by a Sloan Foundation Research Fellowship and the NSF Award CCF-1942583 (CAREER).}\\Computer Science Department\\ Yale University \and Constantinos Daskalakis\footnote{Supported by NSF Awards CCF-1901292, DMS-2022448 and DMS-2134108, by a Simons Investigator Award, by the Simons Collaboration on the Theory of Algorithmic Fairness, by a DSTA grant, and by the DOE PhILMs project (No. DE-AC05-76RL01830). }\\EECS and CSAIL\\ MIT}
\date{}

\begin{document}
\maketitle

\begin{abstract}
        Machine learning has developed a variety of tools for learning and representing high-dimensional distributions with structure. Recent years have also seen big advances in designing multi-item mechanisms. Akin to overfitting, however, these mechanisms can be extremely sensitive to the Bayesian prior that they target, which becomes problematic when that prior is only approximately known. At the same time, even if access to the exact Bayesian prior is given, it is known that optimal or even approximately optimal multi-item mechanisms  run into   sample, computational,  representation and communication intractability barriers.
    
    We consider a natural class of multi-item mechanism design problems with very large numbers of items, but where the bidders' value distributions can be well-approximated by a topic model akin to those used in recommendation systems with very large numbers of possible recommendations. We propose a mechanism design framework for this setting, building on a recent robustification framework by Brustle et al., which disentangles the statistical challenge of estimating a multi-dimensional prior from the task of designing a good mechanism for it, and robustifies the performance of the latter against the estimation error of the former.  We provide an extension of this framework appropriate for our setting, which allows us to exploit the expressive power of topic models to reduce the effective dimensionality of the mechanism design problem and remove the dependence of its computational, communication and representation complexity  on the number of items.
\end{abstract}

\thispagestyle{empty}
\addtocounter{page}{-1}
\newpage


\section{Introduction} \label{sec:intro}

Mechanism Design has found important applications in the design of offline and online markets. One of its main applications is the design of auctions, where a common goal is to maximize the seller's revenue from the sale of one or multiple items to one or multiple bidders. This is challenging because bidders are strategic and interact with the auction in a way that benefits themselves rather than the seller. It is well-understood that, without any information about the bidders' willingness to pay for different bundles of items, there is no meaningful way to optimize revenue. As such, a classical approach in Economics is to assume that bidders' {\em types} -- which determine their values for different bundles  and thus their willingness to pay for different bundles -- are not arbitrary but randomly drawn from a joint distribution $D$ that is common knowledge, i.e.~known to all bidders and the auctioneer. With such a Bayesian prior, the revenue of different mechanisms is compared on the basis of what revenue they achieve in expectation with respect to bidder type vectors drawn from $D$, and assuming that bidders play according to some (Bayesian) Nash equilibrium strategies, or some other type of (boundedly) rational behavior, e.g.~no-regret learning.

Even with a Bayesian prior, however, revenue maximization is quite a challenging task. While Myerson's celebrated work showed that a relatively simple mechanism is optimal in single-item settings~\cite{Myerson81}, characterizing the structure of optimal multi-item mechanisms has been notoriously difficult both analytically and computationally. Indeed, it is known that (even approximately) optimal multi-item mechanisms may require description complexity that scales exponentially in the number of items, even when there is a single buyer~\cite{hart2013menu,dughmiLN14,Daskalakis2017,babaioff2021menu}, they might be computationally intractable, even in simple settings~\cite{CaiDW12b,DaskalakisDT14,chen2015complexity}, and they may exhibit several counter-intuitive properties  which do not arise in single-item settings; see survey~\cite{Daskalakis2015}. Nevertheless, recent years have seen substantial progress on various fronts: analytical characterizations of optimal multi-item mechanisms~\cite{DaskalakisDT13,giannakopoulos2014duality,kash2016optimal,Daskalakis2017}; computational frameworks for computing near-optimal multi-item mechanisms~\cite{AlaeiFHHM12,CaiDW12a,CaiDW12b,CaiDW13a,CaiDW13b}; approximate multi-item revenue optimization via simple mechanisms~\cite{ChawlaHK07,ChawlaHMS10,Alaei11,HartN12,KleinbergW12,CaiH13,BILW14,Yao15,RubinsteinW15,CaiDW16,ChawlaM16,CaiZ17,DaskalakisFLSV20}; and (approximate) multi-item revenue optimization using sample access to the type distribution~\cite{morgenstern2016learning, goldner2016prior, CaiD17,Syrgkanis17,GonczarowskiW18,brustle2020multi}, including via the use of deep learning~\cite{feng2018deep,shen2019automated,dutting2019optimal}.

The afore-described progress on multi-item revenue optimization provides a diversity of tools that can be combined to alleviate the analytical and computational intractability of optimal mechanisms. Yet, there still remains an important challenge in applying those tools, which is that they typically require that the type distribution $D$ is either known or can be sampled. However, this is too strong an assumption.  It is common that $D$ is {\em estimated} through market research or econometric analysis in related settings, involving similar items or a subset of the items. In this case, we would only hope  to know some approximate distribution $\hat D$ that is close to $D$. In
other settings, we may have sample access to the true distribution $D$ but there might be errors in measuring
or recording those samples. Again, we might hope to {\em estimate} an approximate distribution $\hat D$ that is close
to $D$. Unfortunately, it is well understood that designing a mechanism for $\hat D$ and using it for $D$ might be a bad idea, as optimal mechanisms tend to overfit the details of the type distribution. This has motivated a strand of recent literature to study how to robustify mechanisms to errors in the distribution~\cite{BergemannS11,CaiD17,Brustle2019}.

There is, in fact, another important reason why one might want to design mechanisms for some approximate type distribution. Multi-dimensional data is complex and one would want to leverage the extensive statistical and machine learning toolkit that allows approximating such high-dimensional distributions with more structured models. Indeed, while the true type distribution $D$ might not conform to a simple model, it might be close to a distribution $\hat D$ that does. We would like to leverage the simple structure in $\hat D$ to (i) alleviate the computational intractability of multi-item mechanisms, and (ii) reduce the amount of communication that the bidders and the auctioneer need to exchange. While the structured model $\hat D$ might allow (i) and (ii), we need the guarantee that the revenue of our mechanism  be robust when we apply it to the true distribution~$D$.

Motivated by the discussion above, in this work we build a multi-item mechanism design framework that combines matrix factorization models developed for recommendation systems with mechanism design,  targeting two issues: (1) the intractability of Mechanism Design with respect to the number of items (arising from the exponential dependence of the number of types on the number of items if no further assumptions are placed); (2) the lack of exact access to the Bayesian priors. In particular, we assume that each bidder draws their type -- specifying their values for a very large universe of $N$ items (think all restaurants in a city or all items on Amazon) -- from a distribution $D_i$ that is close to a Matrix Factorized model $\hat{D}_i$, whose latent dimension is $k << N$. Targeting these approximate distributions $\hat{D}_i$ allows us to reduce the effective dimensionality of bidder types to $k$, which has huge advantages in terms of the computational/representation/communication/sample complexity of mechanism design. We develop tools that allow us to (a) use the mechanism constructed for the approximate  $\hat{D}_i$'s under the true ${D}_i$'s without sacrificing much revenue; and (b) interact with the bidders who are unaware of the latent codes (they only understand their values for the $N$ items and are oblivious to the matrix factorized model) yet exploit the factorized model for efficiently communicating with them without the impractical burden of having them communicate their $N$-dimensional type to the mechanism. In sum, our results are as follows:

\begin{itemize}
    \item With a query protocol $\query$ that learns an approximate latent representation of a bidder's type, Theorem~\ref{thm:main} shows how to combine it with any mechanism $\givenMech$ that is designed only for the Matrix Factorization model to produce a mechanism that generates comparable revenue but with respect to the true distribution. The result is obtained via a refinement of the robustification result in~\cite{brustle2020multi}, where the loss in revenue, as well as the violation in incentive compatibility  now only depend on the effective dimension of the Matrix Factorization model, $k$, but not the total number of items, $N$ (Lemma~\ref{lem:robustness}).
    
    \item We show that if the valuations are constrained-additive (Definition~\ref{def:constrained additive}), we can obtain communication-efficient query protocols in several natural settings (Theorem~\ref{thm:query protocol}). The queries we consider ask a bidder whether they are willing to purchase an item at a given price. In the first setting, the design matrix of the Matrix Factorization model contains a diagonally dominant matrix -- a generalization of the well-known separability assumption by Donoho and Stodden~\cite{DonohoS03}. In two other settings, we assume that the design matrix is generated from a probablistic model and show that a simple query protocol succeeds  with high probability.
    
    \item Combining Theorems~\ref{thm:main} and~\ref{thm:query protocol}, we show that, given any mechanism $\givenMech$ that is designed only for the Matrix Factorization model, we can design a mechanism that achieves comparable revenue and only requires the bidders to answer a small number of simple queries. In particular, for several natural settings, we show that the number of queries scales \yangnote{quasi-linearly} in the effective dimension of the Matrix Factorization model and independent of the total number of items (Proposition~\ref{prop:1+2 deterministic}).
\end{itemize}
\section{Preliminaries}
\subsection{Brief Introduction to Mechanism Design}
We provide a brief introduction to mechanism design. To avoid a very long introduction, we only define the concepts in the context of multi-item auctions, which will be the focus of this paper. See Chapter 9 of~\cite{AGTbook} and the references therein for a more detailed introduction to mechanism design.

\paragraph{Multi-item Auctions.} The seller is selling \textbf{$N$ heterogenous items} to \textbf{$m$ bidders}. Each bidder $i$ is assumed to have a \textbf{private type} $t_i$ that encodes their preference over the items and bundles of items. We assume that $t_i$ lives in the $N$-dimensional Euclidean space. For each bidder, there is a publicly known valuation function $v_i(\cdot,\cdot)$, where $v_i(t_i,S)\in \mathbb{R}$ is bidder $i$'s value for bundle $S\subseteq [N]$ when $i$'s private type is $t_i$. In this paper, we consider the \emph{Bayesian setting with private types}, that is, each bidder's type $t_i$ is drawn \emph{privately} and \emph{independently} from a publicly known distribution $D_i$.

\vspace{-.1in}
\paragraph{Mechanism.} The seller designs a mechanism to sell the items to bidders. A mechanism consists of an allocation rule and a payment rule, where the allocation rule decides a way to allocate the items to the bidders, and the payment rule decides how much to charge each bidder.

\vspace{-.1in}
\paragraph{Direct Mechanism:} In a {direct mechanism}, the mechanism directly solicits types from the bidders and apply the allocation and payment rules on the reported types. More specifically, for any reported type profile $b=(b_1,\ldots, b_m)$, a direct mechanism $M:=(x(\cdot),p(\cdot))$ selects $x(b)\in\{0,1\}^{m\times N}$ as the allocation and charges each bidder $i$ payment $p_i(b)$.\footnote{Note that $p(b)=(p_1(b),\ldots, p_m(b))$.} We slightly abuse notation to allow the allocation rule to be randomized, so $x(b)\in \Delta\left(\{0,1\}^{m\times N}\right)$. We assume that bidders have quasi-linear utilities. If bidder $i$'s private type is $t_i$, her utility under reported bid profile $b$ is $u_i\left(t_i, M(b)\right)=\E\left[v_i\left(t_i,x(b)\right)-p_i(b)\right]$, where the expectation is over the randomness of the allocation and payment~rule. 

\vspace{-.1in} 
\paragraph{Expected Revenue:} In this paper, our goal is to design mechanisms with high expected revenue. For a direct mechanism $M$, we use $\rev(M,D)$ to denote $\E_{t\sim D}[\sum_{i\in[m]} p_i(t)]$, where $t=(t_1,\ldots, t_m)$ is the type profile and is drawn from $D=\bigtimes_{i\in[m]} D_i$.

\vspace{-.1in}
\paragraph{Incentive Compatibility and Individual Rationality}
Since the bidders' types are private, unless the mechanism \emph{incentivizes} the bidders to report truthfully, there is no reason to expect that the reported types correspond to the true types. The notion of incentive compatibility is defined to capture this.
\begin{itemize}

\item {\bf $\varepsilon$-Bayesian Incentive Compatible ($\varepsilon$-BIC):} if bidders draw their types from some distribution $D=\bigtimes_{i=1}^m D_i$, then a direct mechanism $M$ is \emph{$\varepsilon$-BIC with respect to $D$} if for each bidder $i\in[m]$
$$\mathbb{E}_{t_{-i} \sim D_{-i}}[u_i(t_i,M(t_i,t_{-i}))] \geq \mathbb{E}_{t_{-i} \sim D_{-i}}[u_i(t_i,M(t'_i,t_{-i}))] -\varepsilon,$$ for all potential misreports $t'_i$, in expectation over all other bidders bid $t_{-i}$. A mechanism is BIC if it is $0$-BIC.

\item \yangnote{{\bf $(\varepsilon,\delta)$-Bayesian Incentive Compatible ($(\varepsilon,\delta)$-BIC):} if bidders draw their types from some distribution $D=\bigtimes_{i=1}^m D_i$, then a direct mechanism $M$ is \emph{$(\varepsilon,\delta)$-BIC with respect to $D$} if for each bidder $i\in[m]$:

$$\Pr_{t_i\sim D_i}\left[\mathbb{E}_{t_{-i} \sim D_{-i}}[u_i(t_i,M(t_i,t_{-i}))] \geq \mathbb{E}_{t_{-i} \sim D_{-i}}[u_i(t_i,M(t'_i,t_{-i}))] -\varepsilon\right]\geq 1-\delta.$$ 
}
	\item {\bf Individually Rational (IR):} A direct mechanism $M$ is \emph{IR} if for all type profiles $t=(t_1,\ldots, t_m)$, $$u_i(t_i,M(t_i,t_{-i})) \geq 0$$ for all bidders $i\in [m]$.
	\end{itemize}

\vspace{-.1in}
\paragraph{Indirect Mechanism:} An {\em indirect mechanism} does not directly solicit the bidders' types. After interacting with the bidders, the mechanism selects an allocation and payments. The notions of $\varepsilon$-Bayesian Incentive Compatibility and Individual Rationality can be extended to indirect mechanisms using the solution concept of $\varepsilon$-Bayes Nash equilibrium. \yangnote{The notion of $(\varepsilon,\delta)$-Bayesian Incentive Compatibility can be extended to indirect mechanisms using the new solution concept, which we call {\em $(\varepsilon,\delta)$-weak approximate Bayes Nash equilibrium}. In an incomplete information game, a strategy profile is an $(\varepsilon,\delta)$-weak approximate Bayes Nash equilibrium if for every bidder, with probability no more than $\delta$ (over the randomness of their own type), unilateral deviation from the Bayesian Nash strategy can increase the deviating bidder's expected utility (with respect to the randomness of the other bidders' types and assuming those follow their Bayesian Nash equilibrium strategies) by more than $\varepsilon$.}

\begin{remark}\label{rmk:revenue and IC estimate}
For a $(\varepsilon,\delta)$-weak approximate Bayes Nash equilibrium, its expected revenue computation is made in this paper using the convention that all bidders follow their $(\varepsilon,\delta)$-weak approximate Bayes Nash equilibrium strategies. At a cost of an additive $m^2\delta H$ loss in revenue (where $H$ is the highest possible value of any bidder), we can assume that only the $(1-\delta)$-fraction of types of each bidder who have no more than $\varepsilon$ incentive to deviate from the weak approximate Bayes Nash equilibrium strategies follow these strategies while the remaining $\delta$ fraction use arbitrary strategies. Similarly, we can interpret the $(\varepsilon,\delta)$-weak approximate Bayes Nash equilibrium definition as requiring that at least $(1-\delta)$-fraction of the types of each bidder have at most $O(\varepsilon+m\delta H)$ incentive to deviate from the Bayes Nash strategies assuming that for every other bidder at most $\delta$ fraction of their types deviate from their Bayes Nash strategies.
\end{remark}

\subsection{Further Preliminaries}
\begin{definition}\label{def:prokhorov}
Let $(U,d)$ be a metric space and $\mathcal{B}$ be a $\sigma$-algebra on $U$. For $A \in \mathcal{B}$, let $A^{\varepsilon} = \{x : \exists y \in A \ \  s.t \ \ d(x,y)<\varepsilon \}$. 	Two probability measure 
$P$ and $Q$ on $\mathcal{B}$ have \emph{Prokhorov distance}
$$ \inf \left \{\varepsilon>0 : P(A) \leq Q(A^{\varepsilon}) + \varepsilon \text{ and } \ Q(A)\leq P(A^{\varepsilon})+\varepsilon  ,~ \forall A \in \mathcal{B}\right \}.$$
We consider distributions supported on some Euclidean Space, and we choose {$d$ to be the $\ell_\infty$-distance}. We denote the $\ell_\infty$-Prokhorov distance between distributions $F$, $\widehat{F}$ by $d_P(F,\widehat{F})$.
\end{definition}

We will also make use of the following characterization of the Prokhorov metric by~\cite{Strassen65}.
\begin{lemma}[Characterization of the Prokhorov Metric \cite{Strassen65}]\label{lem:prokhorov characterization}
	Let $F$ and $\widehat{F}$ be two distributions supported on $\mathbb{R}^n$. $d_P(F,\widehat{F})\leq \varepsilon$ if and only if there exists a coupling $\gamma$ of $F$ and $\widehat{F}$, such that $\Pr_{(x,y)\sim \gamma}\left[\normI{x-y}>\varepsilon \right]\leq~\varepsilon$.
	\end{lemma}

\begin{definition}[Influence Matrix and Weak Dependence]
For any $d$-dimensional random vector ${X}=(X_1,\ldots, X_d)$, we define the influence of variable $j$ on variable $i$ as $$\alpha_{i,j}:=\sup_{\substack{x_{-i-j} \\x_j\neq x'_j}} d_{TV}\left(F_{X_i\mid X_j=x_j, X_{-i-j}=x_{-i-j}},F_{X_i\mid X_j=x'_j, X_{-i-j}= x_{-i-j}}\right), $$ where $F_{X_i\mid X_{-i}=x_{-i}}$ denotes the conditional distribution of $X_i$ given $X_{-i}=x_{-i}$, and $d_{TV}(D,D')$ denotes the total variational distance between distribution $D$ and $D'$.
Also, let $\alpha_{i,i}:=0$ for each $i$, and we use $\Influ{X}$ to denote the $d\times d$ matrix $(\alpha_{i,j})_{i\in[d],j\in[d]}$. In this paper, we consider the coordinates of $X$ to be \textbf{weakly dependent} if $\norm{\Influ{X}}<1$.
\end{definition} 

\section{Our Model and Main Results}

\paragraph{Setting and Goal:} We consider a classical mechanism design problem, wherein a seller is selling $N$ items to $m$ buyers, where buyer $i$'s type $t_i$ is drawn from a distribution $D_i$ over $\mathbb{R}^N$ independently. The goal is to design a mechanism that maximizes the seller's revenue. In this paper, we operate in a setting where $D_i$ is unknown, but we are given access to the following components: (I) For each bidder~$i$, we are given a machine learning model $\mlDistI$ --- of the matrix factorization type as described below, which approximates $D_i$. (II) We are given a good mechanism $\givenMech$ for the approximate type distributions; in its design this mechanism can exploit the low effective dimensionality, $k$, of types in the approximate model. Our goal is (III)  to use (I) and (II) to obtain a good mechanism for the true type distributions. 

\paragraph{(I) The Machine Learning Component:} We assume that each bidder's type distribution~$D_i$ can be well-approximated by a known \emph{
Matrix Factorization 
(MF)} model $\mlDistI$. In particular: 
\begin{itemize}
\item We use $A\in \mathbb{R}^{N\times k}$ to denote the design matrix of the model, where each column can be viewed as the type (over $N$ items) of an ``archetype.'' As described in the following two bullets, types are sampled by each $\mlDistI$ as linear combinations over archetypes. 
\item We use $\latentDistI$ to denote a distribution over $[0,1]^k$. The subscript $z$ is not a parameter of the distribution --- it serves to remind us that this distribution samples in the latent space $[0,1]^k$ and distinguish it from the distribution $\mlDistI$ defined next.
\item If $F$ is a distribution over $\mathbb{R}^k$, we use $A\circ F$ to denote the distribution of the random variable $Az$, where $z\sim F$. With this notation, we  use  $\mlDistI$ to denote $A\circ \latentDistI$.
\item We assume that, for each bidder, the matrix factorization model is not far away from the true type distribution, that is, for some $\varepsilon_1 >0$ we have that  $d_P(D_i,\mlDistI) \le \varepsilon_1$ for all $i\in [m]$.
\end{itemize}

\begin{remark}
In the above description we assumed that all $\mlDistI$'s share the same design matrix $A$. This is done to avoid overloading notation but all our results would hold if each $\mlDistI$ had its own design matrix $A_i$.
\end{remark}

\paragraph{(II) The Mechanism Design Component:} We assume that we are given a direct mechanism $\givenMech$ for types drawn from the Machine Learning model. In particular, we assume that this mechanism makes use of the effective dimension $k$ of the Machine Learning model, accepting ``latent types'' (of dimension $k$) as input from the bidders. Specifically:
\begin{itemize}
    \item Recall that, for each bidder $i$, their valuation function $v_i(\cdot,\cdot): \mathbb{R}^N\times 2^{[N]}\rightarrow \mathbb{R}$ is common knowledge. (Recall that $v_i$ takes as input the bidder's type and a subset of items so how the bidder values different subsets of items depends on their private type.)
    
    \item\label{itm:induced value}  The designer is given $A$ and $\latentDistI$ for each bidder $i$, and treats bidder $i$'s type as drawn from $\latentDistI$, i.e.~in the latent space $[0,1]^k$. With respect to such ``latent types,'' there is an induced valuation function. In particular, for each bidder $i$, we use $v^A_i:\mathbb{R}^k\times 2^{[N]} \rightarrow \mathbb{R}$ to denote the valuation function defined as follows ${v}^A_i(z_i,S):=v_i(Az_i,S)$, where $z_i\in\mathbb{R}^k$.
    \item With the above as setup, we assume that the designer designs a mechanism $\givenMech$ that is BIC and IR w.r.t. $\latentDist=\bigtimes_{i=1}^m \latentDistI$ and valuation functions $\{v^A_i(\cdot,\cdot)\}_{i\in[m]}$.
\end{itemize}


\paragraph{(III) The New Component:} We consider the regime where $N \gg k$, and our goal is to combine the Machine Learning component with the Mechanism Design component to produce a mechanism which generates revenue comparable to $\rev(\givenMech,\latentDist)$ when used for bidders whose types are drawn from $D=\bigtimes_{i=1}^m D_i$. There are two challenges: (i)  $\givenMech$ takes as input the latent representation of a bidder's type under $\latentDist$, however under $D$ a bidder is simply ignorant about any latent representation of their type so they cannot be asked about it;
(ii)  $\givenMech$'s revenue is evaluated with respect to $\latentDist$ and valuation functions $\{v^A_i(\cdot,\cdot)\}_{i\in[m]}$ and our goal is to obtain a mechanism whose revenue is similar  under $D$ and valuation functions $\{v_i(\cdot,\cdot)\}_{i\in[m]}$. We show how to use a communication efficient query protocol together with a robustification procedure to combine the Machine Learning and Mechanism Design components. 

To state our results, we first need to formally define  query protocols and some of their properties.

\begin{definition}[$(\varepsilon,\delta)$-query protocol]\label{def:query protocol}
Let $\query$ be a \textbf{query protocol}, i.e., some communication protocol that exchanges messages with a bidder over possibly several rounds and outputs a vector in $\mathbb{R}^k$. We say that a bidder interacts with the query protocol truthfully, if whenever the protocol asks the bidder to evaluate some function on their type the bidder evaluates the function and returns the result truthfully. We use  $\query(t)\in \mathbb{R}^k$ to denote the output of $\query$ when interacting with a {truthful bidder} whose type is $t\in\mathbb{R}^N$. $\query$ is called a $(\varepsilon,\delta)$-query protocol, if for any $t\in\mathbb{R}^N$ and $z\in \mathbb{R}^k$  satisfying $\normI{t-Az}\leq \varepsilon$, we have that $\normI{z-\query(t)} \le \delta$. 
\end{definition}

We also need the notion of Lipschitz valuations to formally state our result.
\begin{definition}[Lipschitz Valuations]\label{def:Lipschitz}
   $v(\cdot,\cdot):\mathbb{R}^N\times 2^{[N]}\rightarrow \mathbb{R}$ is a $\lipshitz$-Lipschitz valuation, if for any two types $t,t'\in\mathbb{R}^N$ and any bundle $S\subseteq [N]$, $|v(t,S)-v(t',S)|\leq \lipshitz\normI{t-t'}$. 
\end{definition}

\noindent This includes familiar settings, for example if the bidder is $c$-demand, the Lipschitz constant $\lipshitz=c$.\footnote{A bidder is $c$-demand if for any set $S$ of items, the bidder picks their favorite bundle with size no more than $c$ in $S$ evaluating the value of each such bundle additively, with values as determined by the bidder's type $t$. Formally, $v(t,S)=\max_{B\subseteq S, |B|\leq c} \sum_{j\in B} t_j$.} 

We are now ready to state our first main result.

\begin{theorem}\label{thm:main}
Let $D=\bigtimes_{i=1}^m D_i$ be the bidders' type distributions and $v_i:\mathbb{R}^N\times 2^{[N]}\rightarrow \mathbb{R}$ be a $\lipshitz$-Lipschitz valuation for each bidder $i\in[m]$. Also, let $A\in \mathbb{R}^{N\times k}$ be a design matrix and $\latentDistI$ be a distribution over $\mathbb{R}^k$ 
for each $i\in[m]$. 

Suppose we are given query access to a mechanism $\givenMech$ that is BIC and IR w.r.t. $\latentDist=\bigtimes_{i=1}^m \latentDistI$ and valuations $\{{v}^A_i\}_{i\in[m]}$ (as defined in the second bullet of the Mechanism Design component above), and there exists $\varepsilon_1>0$ such that $d_P(D_i,A\circ \latentDistI)\leq \varepsilon_1$  for all $i\in[m]$. Given any $(\varepsilon_1,\varepsilon)$-query protocol with $\varepsilon\geq \varepsilon_1$, we can construct mechanism $M$  using only query access to $\givenMech$ and 
obliviously with respect to $D$, such that for any possible $D$ that satisfies the above conditions of Prokhorov distance closeness the following hold:
	 	\begin{enumerate}
	 	\item $M$ only interacts with every bidder using $\query$ once;
	 		\item $M$ is \yangnote{ $(\kappa,\varepsilon_1)$-BIC} w.r.t. $D$ and IR, where \yangnote{$\kappa=O\left( \lipshitz  \varepsilon_1 +\normI{A}\lipshitz  m \varepsilon+  \normI{A}\lipshitz \sqrt{m\varepsilon}\right)$};
	 		\item  The expected revenue of $M$ is at least $ \rev(\givenMech,\latentDist)-O\left(m\kappa\right).$
	 	\end{enumerate}



\end{theorem}

\begin{remark}
The mechanism $M$ will be an indirect mechanism. We are slightly imprecise here to call the mechanism $(\kappa,\varepsilon_1)$-BIC. Formally what we mean is that interacting with $\query$ \emph{truthfully} is a $(\kappa,\varepsilon_1)$-weak approximate Bayes Nash equilibrium. We compute the expected revenue assuming all bidders interacting with $\query$ \emph{truthfully}. As we discussed in~\Cref{rmk:revenue and IC estimate}, with an additional additive $\normI{A} \lipshitz m^2 \varepsilon_1$ loss in revenue, we can assume that only the $(1-\delta)$-fraction of types of each bidder who have no more than $\varepsilon$ incentive to deviate from the Bayes Nash strategies interact with $\query$ truthfully while the remaining $\delta$ fraction uses arbitrary strategies.

\end{remark}

\paragraph{\textbf{Why isn't ~\cite{brustle2020multi} sufficient?}} One may be tempted to prove  Theorem~\ref{thm:main} using~\cite{brustle2020multi}. However, there are two subtle issues with this approach: (i) The violation of the incentive compatibility constraints and the revenue loss of the robustification process in~\cite{brustle2020multi} depend linearly in $N$, rather than in $\normI{A}$ as in Theorem~\ref{thm:main}. Note that $\normI{A}=\max_{i\in[N]} \sum_{j=1}^k |A_{ij}|$, which only depends on $k$ and the largest value an archetype can have for a single item and thus could be significantly smaller than $N$. (ii) The robustification process involves sampling from the conditional distribution of $A\circ \latentDistI$ on an $N$-dimensional cube, which is equivalent to sampling from the conditional distribution of $\latentDistI$ on a set $S$ whose image after the linear transformation $A$ is the $N$-dimensional cube. However, $S$ may be difficult to sample from if $A$ is not a well-conditioned.

In the following lemma, we refine the robustification result in~\cite{brustle2020multi} (Theorem 3 in that paper) and show that  given an approximate distribution $\widehat{F}$ in the latent space and a BIC and IR mechanism $\widehat{M}$ w.r.t. $\widehat{F}$, we can \emph{robustify} $\widehat{M}$  with \emph{negligible revenue loss} so that it is an approximately BIC and exactly IR mechanism w.r.t. $F$ for any distribution $F$ that is within the $\varepsilon$-Prokhorov ball around $\widehat{F}$. Importantly, we exploit the effective dimension of the matrix factorization model to replace the dependence on $N$ with $\normI{A}$ in both the violation of the incentive compatibility constraints and the revenue loss. Additionally, we only need to be able to sample from the conditional distribution of $\latentDistI$ on a $k$-dimensional cube. We postpone the proof of Lemma~\ref{lem:robustness} to the Appendix~\ref{sec:robustness proof}.

\begin{lemma}\label{lem:robustness}
	 	Let $A\in \mathbb{R}^{N\times k}$ be the design matrix. Suppose we are given a collection of distributions over latent types $\{\widehat{F}_{z,i}\}_{i\in[m]}$, where the support of each $\widehat{F}_{z,i}$ lies in $[0,1]^k$, and a BIC and IR mechanism $\widehat{M}$ w.r.t.~$\widehat{F}=\bigtimes_{i=1}^m 
	 	\widehat{F}_{z,i}$ and valuations $\{v^A_i\}_{i\in [m]}$, where each $v_i$ is an $\lipshitz$-Lipschitz valuation. Let $F=\bigtimes_{i=1}^m 
	 	F_{z,i}$ be any distribution such that $d_P(F_{z,i},\widehat{F}_{z,i})\leq \varepsilon$ for all $i\in[m]$. Given access to a \textbf{sampling algorithm $\condsample_i$} for each $i\in[m]$, where $\condsample_i(x,\delta)$ draws a sample from the conditional distribution of $\widehat{F}_{z,i}$ on the $k$-dimensional cube $\bigtimes_{j\in [k]}[x_j,x_j+\delta)$, we can construct a randomized mechanism $\widetilde{M}$ using only query access to $\givenMech$ and 
	 	obliviously with respect to $F$, such that for any $F$ satisfying the above conditions of Prokhorov distance closeness the following hold: 
	 	\begin{enumerate}
	 		\item $M$ is  
	 		$\kappa$-BIC and IR w.r.t.~ $F$ and valuations $\{v^A_i\}_{i\in [m]}$, where  $\kappa=O\left( \normI{A}\lipshitz  m \varepsilon+\normI{A}\lipshitz \left(\delta+ \frac{m\varepsilon}{\delta} \right)\right)$;
	 		\item  The expected revenue of $\widetilde{M}$ is $\rev\left(\widetilde{M},F\right)\geq 	 		\rev(\widehat{M},\widehat{F})-O\left(m\kappa\right).$
	 	\end{enumerate}
\end{lemma}

Equipped with Lemma~\ref{lem:robustness}, we proceed to prove Theorem~\ref{thm:main}.

\begin{prevproof}{Theorem}{thm:main}
Consider the following mechanism:
\begin{algorithm}[h]
\begin{algorithmic}[1]
\STATE Construct mechanism $\widetilde{M}$ using Lemma~\ref{lem:robustness} by choosing $\widehat{F}_{z,i}$ to be $\latentDistI$ for each $i\in [m]$
and $\delta$ to be $\sqrt{m\varepsilon}$.
\STATE Query each agent $i$ using $\query$. Let $\query(b_i)$ be the output after interacting with bidder $i$.~(For any possible output produced by $\query$, there exists a type $b\in \mathbb{R}^N$, so this is w.l.o.g..) 
\STATE Execute mechanism $\widetilde{M}$ on bid profile $\left(\query(b_1),\ldots, \query(b_m)\right)$.
\end{algorithmic}
\caption{Query-based Indirect Mechanism $M$}
\label{mechanism}
\end{algorithm}

Let $t_i$ be bidder $i$'s type and $z_i$ be a random variable distributed according to $\latentDistI$. Since $d_P(D_i,\mlDistI)\leq \varepsilon_1$, Lemma~\ref{lem:prokhorov characterization} guarantees a coupling between $t_i$ and $Az_i$ such that their $\ell_\infty$ distance is more than $\varepsilon_1$ with probability no more than $\varepsilon_1$. As $\query$ is a $(\varepsilon_1,\varepsilon)$-query protocol, when $t_i$ and $Az_i$ are not $\varepsilon_1$ away, we have $\normI{\query(t_i)-z_i}\leq \varepsilon$. Hence, there exists a coupling between $\query(t_i)$ and $z_i$ so that their $\ell_\infty$ distance is more than $\varepsilon$ with probability no more than $\varepsilon$ (recall $\varepsilon_1\leq \varepsilon$). If we choose $F_{z,i}$ to be the distribution of $\query(t_i)$, $\widehat{F}_{z,i}$ to be $\latentDistI$, and $\delta$ to be $\sqrt{m\varepsilon}$,~\Cref{lem:robustness} states that $\widetilde{M}$ is a $O\left( \normI{A}\lipshitz  m \varepsilon+  \normI{A}\lipshitz \sqrt{m\varepsilon}\right)$-BIC mechanism if bidder $i$ has valuation $v_i^A(\cdot)$ and type $\query(t_i)$. \yangnote{Consider two cases: (a) When $\normI{t_i-Az_i}\leq \varepsilon_1$, then $\normI{t_i-A\query(t_i)}\leq \varepsilon_1+\normI{A} \varepsilon$. Since $v_i(\cdot)$ is $\lipshitz$-Lipschitz, deviating from interacting with $\query$ truthfully can increase the expected utility by at most $O\left(\lipshitz\varepsilon_1+\normI{A}\lipshitz  m \varepsilon+  \normI{A}\lipshitz \sqrt{m\varepsilon}\right)$. (b) When $\normI{t_i-Az_i}>\varepsilon_1$, the bidder may substantially improve their expected utility by deviating. Luckily, such case happens with probability no more than $\varepsilon_1$. }
\end{prevproof}

In Theorem~\ref{thm:query protocol}, we show how to obtain $(\varepsilon,\delta)$-queries under various settings. We further assume that the bidders' valuations are all constrained-additive.

\begin{definition}[Constrained-Additive valuation]\label{def:constrained additive}
A valuation function $v:\mathbb{R}^N\times2^{[N]}\rightarrow \mathbb{R}$ is constrained additive if $v(t,S)=\max_{T\in \mathcal{I}\cap 2^S} \sum_{j\in T} (\mu_j+t_j)$, where $\mathcal{I}$ is a downward-closed set system, and $\mu=(\mu_1,\ldots,\mu_N)$ is a fixed vector.\footnote{One can interpret $\mu$ as the common based values for the items that are shared among all types.} For example, unit-demand valuation is when $\mathcal{I}$ includes all subsets with size no more than $1$. If all elements of $\mathcal{I}$ have size no more than $\lipshitz$, then $v$ is a $\lipshitz$-Lipschitz valuation.
\end{definition}

\begin{theorem}\label{thm:query protocol} 
Let all bidders' valuations be constrained-additive. We consider queries of the form: $e_j^T t\overset{?}{\geq} p$, where $e_j$ is the $j$-th standard unit vector in $\mathbb{R}^N$. The query simply asks whether the bidder is willing to pay at least $p+\mu_j$ for winning item $j$. The bidder provides a \emph{Yes/No} answer. We obtain communicationally efficient protocols in the following settings: 
\begin{itemize}
    \item \textbf{Deterministic Structure:} If $A^T$ can be expressed as $[C^T H^T]\Pi_N$, where $\Pi_N\in \mathbb{R}^N$ is a permutation matrix, $H$ is an arbitrary $(N-k)\times k$ matrix, and $C\in \mathbb{R}^{k\times k}$ is diagonally dominant both by rows and by columns.  This is a relaxation of the well-known \textbf{separability assumption} by Donoho and Stodden~\cite{DonohoS03}, that is, $A^T$ can be expressed as $[I_k H^T]\Pi_N$, where $I_k$ is the $k$-dimensional identity matrix. Let $\alpha=\min_{i\in[k]} \left(|C_{ii}|-\sum_{j\neq i} |C_{ij}|\right)$ and $\beta = \min_{j\in[k]} \left(|C_{jj}|-\sum_{i\neq j} |C_{ij}|\right)$. We have a $\left(\varepsilon,\frac{4 \cdot \max_{j\in[k]} C_{jj} }{\alpha\beta}\cdot \varepsilon\right)$-query protocol using $O\left(k\cdot\log\left(\frac{\normI{A}}{\varepsilon} \right)\right)$ queries for any $\varepsilon>0$.
    \item \textbf{Ex-ante Analysis:} If $A$ is generated from a distribution, where each archetype is an independent copy of a $N$-dimensional random vector $\theta$.
    \begin{itemize}
        \item \textbf{Multivariate Gaussian Distributions:} $\theta$ is distributed according to a multivariate Gaussian distribution $\normal(0,\Sigma)$. If there exists a subset $S\subseteq [N]$ such that $\frac{\Tr{\Sigma_S}}{\rho(\Sigma_S)}> 64k$, where $\Sigma_S=\E[\theta_S\theta_S^T]$ is the covariance matrix for items in $S$ and $\rho(\Sigma_S)$ is the largest eigenvalue of $\Sigma_S$,\footnote{$\theta_S$ is the $|S|$-dimensional vector that contains all $\theta_i$ with $i\in S$.} then with probability at least $1-2\exp\left(-\frac{\Tr{\Sigma_S}}{16\cdot \rho(\Sigma_S)}\right)$, we have a $\left(\varepsilon,\frac{64\sqrt{|S| k}}{\sqrt{\Tr{\Sigma_S}}} \cdot \varepsilon\right)$-query protocol using $O\left(|S|\cdot\log\left(\frac{\normI{A}}{\varepsilon} \right)\right)$ queries for any $\varepsilon>0$. Note that when the entries of $\theta$ are i.i.d., any $S$ with size at least $64k$ satisfies the condition. 
        \item \textbf{Bounded Distributions with Weak Dependence:} Let $\theta_i$ be supported on $[-c,c]$ and has mean $0$ for each $i\in[N]$. If there exists a subset $S\subseteq [N]$ such that $\norm{\Influ{\theta_S}}<1$, and $\sum_{i\in S}v_i^2>\frac{16c^2 k\sqrt{|S|}}{1-\norm{\Influ{\theta_S}}}$, where $v_i^2:=\var[\theta_i]$, then with probability at least \\ $1-2\exp\left(-\frac{\left(1-\norm{\Influ{\theta_S}}\right)\cdot (\sum_{i\in S}v_i^2)^2}{64c^4 k|S|}\right)$, we have a $\left(\varepsilon,\frac{64\sqrt{|S| k}}{\sqrt{\sum_{i\in S}v_i^2}} \cdot \varepsilon\right)$-query protocol using \\ $O\left(|S|\cdot\log\left(\frac{\normI{A}}{\varepsilon} \right)\right)$  queries for any $\varepsilon>0$. Note that when the entries of $\theta$ are independent, $\norm{\Influ{\theta_S}}=0$ for any set $S$. If each $\theta_i$ has variance $\Omega(c^2)$, then any set with size at least $\alpha k^2$ suffices for some absolute constant $\alpha$.
    \end{itemize}
\end{itemize}
\begin{remark}
In the ex-ante analysis, the success probabilities depend on the parameters of the distributions, but note that they are both at least $1-2\exp(-4k)$.
\end{remark}
\end{theorem}

Before we prove Theorem~\ref{thm:query protocol}, we combine it with Theorem~\ref{thm:main} to derive results for a few concrete settings.
\begin{proposition}\label{prop:1+2 deterministic}
Under the same setting as in Theorem~\ref{thm:main} with the extra assumption that every valuation $v_i$ is constrained-additive, 
we can construct mechanism $M$  using only query access to the given mechanism $\givenMech$ and 
oblivious to the true type distribution $D$, such that for any possible $D$, \yangnote{$M$ is
$\left(\eta,\varepsilon_1\right)$-BIC and IR, where $\eta =O\left(\lipshitz\varepsilon_1+ \normI{A}\lipshitz  m f(\varepsilon_1)+  \normI{A}\lipshitz \sqrt{m f(\varepsilon_1)}\right)$,} and has revenue at least $\rev(\widehat{M},\widehat{D}_z)-O\left(\normI{A}\lipshitz  m^2 f(\varepsilon_1)+  \normI{A}\lipshitz m^{3/2}f(\varepsilon_1)^{1/2}\right)$. Recall that $\varepsilon_1$ satisfies $d_P(D_i,A\circ \latentDistI)\leq \varepsilon_1$ for all $i\in [m]$. We compute the function $f(\cdot)$ and the number of queries for the following three concrete settings (one for each of the three assumptions in Theorem~\ref{thm:query protocol}).
\begin{enumerate}
    \item \textbf{Deterministic Structure: Separability.} If the design matrix $A$ satisfies the \textbf{separability assumption} by Donoho and Stodden~\cite{DonohoS03}, that is, $A^T$ can be expressed as $[I_k H^T]\Pi_N$, where $\Pi_N\in \mathbb{R}^N$ is a permutation matrix, $f(\varepsilon_1)=4\varepsilon_1$ for all $\varepsilon>0$. The number of queries each bidder needs to answer is $O\left(k\cdot\log\left(\frac{\normI{A}}{\varepsilon_1} \right)\right)$ .
    \item \textbf{Multivariate Gaussian Distributions: Well-Conditioned Covariance Matrix.} Let $A$ be generated from a distribution, where each archetype is an independent draw from a $N$-dimensional normal distribution $\normal(0,\Sigma)$. Let $\kappa(\Sigma)$ be the condition number of $\Sigma$.\footnote{$\Sigma$ is well-conditioned if $\kappa(\Sigma)$ is small. When $\Sigma=I_N$, $\kappa(\Sigma)=1$.} For any set $S$ with size $64\kappa(\Sigma)k$, if we query each bidder about items in $S$, with probability at least $1-2\exp(-4k)$, $f(\varepsilon_1)=O\left(\frac{k\sqrt{\kappa(\Sigma)}}{\sqrt{\Tr{\Sigma_S}}}\cdot \varepsilon_1 \right)$, and each bidder needs to answer $O\left(\kappa(\Sigma)k\cdot\log\left(\frac{\normI{A}}{\varepsilon_1} \right)\right)$ queries.
    \item \textbf{Weak Dependence: Sufficient Variance per Item.}  Let $A$ be generated from a distribution, where each archetype is an independent copy of an $N$-dimensional random vector $\theta$. Assuming (i) $\norm{\Influ{\theta}}<1$, (ii) $\theta_i$ lies in $[-c,c]$, and (iii) $\var[\theta_i]\geq a^2$ for each $i\in[N]$, then for any set $S$ with size $\frac{256 c^4 k^2}{a^4\left(1-\norm{INF(\theta)}\right)^2}$, if we query each bidder about items in $S$, with probability at least $1-2\exp(-4k)$, $f(\varepsilon_1)=O\left(\frac{\sqrt{k}}{a}\cdot \varepsilon_1\right)$ and each bidder needs to answer $O\left( \frac{c^4 k^2}{a^4\left(1-\norm{INF(\theta)}\right)^2}\cdot \log\left(\frac{\normI{A}}{\varepsilon_1} \right)\right)$ queries.\footnote{Clearly, we can weaken condition (i),(ii) and (iii). The result still holds if we can find a set $S$, so that for vector $\theta_S$, condition (i), (ii), and (iii) hold, and $|S|$ is at least $\frac{256 c^4 k^2}{a^4\left(1-\norm{INF(\theta_S)}\right)^2}$.}
\end{enumerate}
\end{proposition}

\begin{proof}
The results in the first and last setting follows directly from Theorem~\ref{thm:query protocol}. For the second setting, notice that by the eigenvalue interlacing theorem, $\kappa(\Sigma_S)\leq \kappa(\Sigma)$, as $\Sigma_S$ is a principal submatrix of $\Sigma$. Therefore, $\frac{\Tr{\Sigma_S}}{\rho(\Sigma_S)}\geq \frac{|S|}{\kappa(\Sigma_S)}\geq 64 k$. Now, the result follows from Theorem~\ref{thm:query protocol}.
\end{proof}

\begin{prevproof}{Theorem}{thm:query protocol}
Instead of directly studying the query complexity under our query model. We first consider the query complexity under a seemingly stronger query model, where we directly query the bidder about their value of $e_j^T t$, and their answer will be within $e_j^T t\pm \eta$ for some $\eta>0$. We refer to this type of queries as noisy value queries. Since for each item $j$, $|e_j^T Az|\leq \normI{A}$ for all $z\in [0,1]^k$ and we only care about types in $\mathbb{R}^N$ that are close to some $Az$, we can use our queries to perform binary search on $p$ to simulate noisy value queries. In particular, we only need $\log{\normI{A}}+\log{1/\eta}+\log{1/\varepsilon}$ many queries to simulate one noisy value queries. From now on, the plan is to first investigate the query complexity for noisy value queries, then convert the result to query complexity in the original model. 

We first fix the notation. Let $\ell$ be the number of noisy value queries, and $Q\in \mathbb{R}^{\ell\times N}$ be the query matrix, where, each row of $Q$ is a standard unit vector. We use $\qanswer\in \mathbb{R}^\ell$ to denote the bidder's answers to the queries and $\tanswer\in\mathbb{R}^\ell$ to true answers to the queries. Note that $\normI{\qanswer-\tanswer}\leq \eta$. Given $\qanswer$, we solve the following least squares problem: $\min_{z\in\mathbb{R}^k}\norm{QAz-\qanswer}^2$.

The problem has a closed form solution: $\hat{z}=\left(A^TQ^TQA\right)^{-1}A^TQ^T\qanswer$. Let $B:=QA$, and $z(t)\in \mathbb{R}^k$ be a vector that satisfies $\normI{t-Az(t)}\leq \varepsilon$. We are interested in upper bounding $\normI{\hat{z}-z(t)}$. Note that
\begin{align*}
    \hat{z}-z(t) =& (B^TB)^{-1}B^T (\qanswer-Bz(t))\\
    =& (B^TB)^{-1} B^T \left((\qanswer-\tanswer)+(\tanswer-Bz(t))\right)\\
    =& (B^TB)^{-1} B^T (\qanswer-\tanswer) + (B^TB)^{-1} B^T Q(t-Az(t))
\end{align*}

Since the rows of $Q$ are all standard unit vectors, $\normI{Q}=1$. 
\begin{align*}
   \normI{\hat{z}-z(t)}\leq &\normI{(B^TB)^{-1} B^T (\qanswer-\tanswer)} +\normI{(B^TB)^{-1} B^T Q(t-Az(t))}\\
   \leq &\normI{(B^TB)^{-1}}\normI{B^T}\left(\eta+\normI{Q(t-Az(t))}\right)\\
   \leq &\normI{(B^TB)^{-1}}\normI{B^T}(\eta+\varepsilon).
\end{align*} Next, we bound $\normI{(B^TB)^{-1}}\normI{B^T}$ under the different assumptions.

\paragraph{Deterministic Structure:} We choose $\ell=k$ and $Q$ so that $QA=B=C$. Since $C$ is diagonally dominant, $C$ is non-singular, and $(C^T C)^{-1} = C^{-1}(C^{T})^{-1}$.

\begin{lemma}[Adapted from Theorem 1 and Corollary 1 of~\cite{varah1975lower}]\label{lem:diagonally dominant l-infinity}
         If a matrix $U\in \mathbb{R}^{n\times n}$ is diagonally dominant both by rows and by columns, and $\alpha=\min_{i\in[n]} \left(|U_{ii}|-\sum_{j\neq i} |U_{ij}|\right)$ and $\beta = \min_{j\in[n]} \left(|U_{jj}|-\sum_{i\neq j} |U_{ij}|\right)$, then $\normI{U^{-1}}\leq 1/\alpha$ and $\normI{(U^T)^{-1}}\leq 1/\beta$.
\end{lemma}

By Lemma~\ref{lem:diagonally dominant l-infinity}, $\normI{(C^TC)^{-1}}\normI{ C^T} \leq \frac{\normI{C^T}}{\alpha\beta}$. Note that $\normI{C^T}=\max_{j\in[k]} \sum_{i\in[k]} |C_{ij}|\leq 2\max_{j\in[k]} C_{jj}$. The last inequality is because $C$ is diagonally dominant by columns. To sum up, if we choose $Q$ so that $QA=C$, $$\normI{\hat{z}-z(t)}\leq \frac{ (\varepsilon+\eta)\cdot\normI{C^T}}{\alpha\beta}\leq \frac{2 (\varepsilon+\eta)\cdot \max_{j\in[k]} C_{jj} }{\alpha\beta}.$$

\paragraph{Ex-ante Analysis:} 

Since $\normI{(B^TB)^{-1}}\leq \sqrt{k}\norm{(B^TB)^{-1}}$ and $\normI{B^T}\leq \sqrt{\ell}\norm{B}$, $$\normI{\hat{z}-z(t)}\leq  \frac{\sqrt{\ell k}\cdot \sigma_{max}(B)}{\sigma_{min}(B)^2}\cdot(\eta+\varepsilon),$$ where $\sigma_{max}(B)$ (or $\sigma_{min}(B)$) is $B$'s largest (or smallest) singular value. 

\paragraph{Multivariate Gaussian distribution:} When $\theta$ is distributed according to a multivariate Gaussian distribution, we choose $\ell=|S|$ and $Q$ so that each row corresponding to an $e_j$ with $j\in S$. Now, $B$ is a $\ell\times k$ random matrix where each column is an independent copy of $\theta_S$.  We use Lemma~\ref{lem:GGM singular values} to  bound $B$'s largest singular value  $\lsv{B}$ and smallest singular value $\ssv{B}$. The proof of Lemma~\ref{lem:GGM singular values} is postponed to Section~\ref{sec:proof of concentration Gaussian}.

\begin{lemma}\label{lem:GGM singular values}[Concentration of Singular Values under multivariate Gaussian distributions]\\
         Let $U=[X^{(1)},\ldots, X^{(n)}]$ be a $m\times n$ random matrix, where each column of $U$ is an independent copy of a $m$-dimensional random vector $X$ distributed according to a multivariate Gaussian distribution $\normal(0,\Lambda^T D\Lambda)$. In particular, $\Lambda\in \mathbb{R}^{m\times m}$ is an orthonormal matrix, and $D\in \mathbb{R}^{m\times m}$ is a diagonal matrix. We have $\sigma_{max}(U)\leq 2\sqrt{\Tr{D}}  \text{ and }\sigma_{min}(U)\geq \frac{\sqrt{\Tr{D}}}{4},$ with probability at least $1-2\exp\left(-\frac{\Tr{D}}{8\cdot d_{max}}+4n 
         \right)$, where $d_{max}$ is the largest entry in $D$. 
\end{lemma}

Since $\frac{\Tr{\Sigma_S}}{\rho(\Sigma_S)}>64 k$, by Lemma~\ref{lem:GGM singular values}, $\lsv{B}\leq 2\sqrt{\Tr{\Sigma_S}}$ and $\ssv{B}\geq \sqrt{\Tr{\Sigma_S}}/4$ with probability at least $1-2\exp\left(-\frac{\Tr{\Sigma_S}}{16\cdot \rho(\Sigma_S)}\right)\geq 1-2\exp(-4k)$. Hence, $\normI{\hat{z}-z(t)}\leq \frac{32\sqrt{|S| k}}{\sqrt{\Tr{\Sigma_S}}}\cdot (\eta+\varepsilon)$ with probability at least $1-2\exp\left(-\frac{\Tr{\Sigma_S}}{16\cdot \rho(\Sigma_S)}\right)$.

\paragraph{Weakly Dependent Distributions:} When the coordinates of $\theta_S$ are weakly dependent, i.e., $\norm{\Influ{\theta_S}}<1$, we choose $\ell=|S|$ and $Q$ so that each row corresponding to an $e_j$ with $j\in S$. Now, $B$ is a $\ell\times k$ random matrix where each column is an independent copy of $\theta_S$. We use Lemma~\ref{lem:singular values under Dobrushin} to  bound  $B$'s largest singular value $\lsv{B}$ and smallest singular value $\ssv{B}$. The proof of Lemma~\ref{lem:singular values under Dobrushin} is postponed to Section~\ref{sec:proof of concentration Dobrushin}.

\begin{lemma}\label{lem:singular values under Dobrushin}[Concentration of Singular Values under Weak Dependence]\\
            Let $U=[X^{(1)},\ldots, X^{(n)}]$ be a $m\times n$ random matrix, where each column of $U$ is an independent copy of a $m$-dimensional random vector $X$. We assume that the coordinates of $X$ are weakly dependent, i.e., $\norm{\Influ{X}}<1$, and each coordinate of $X$ lies in $[-c,c]$ and has mean $0$ and variance $v_i^2$. Let $v=\sqrt{\sum_{i\in[m]} v_i^2}$.  
            We have $\sigma_{max}(U)\leq 2v  \text{ and }\sigma_{min}(U)\geq \frac{v}{4},$ with probability at least $1-2\exp\left(-\frac{\left(1-\norm{\Influ{X}}\right)v^4}{32c^4 nm}+4 n
         \right)$.
\end{lemma}

Since $\sum_{i\in S}v_i^2>\frac{16 c^2 k\sqrt{|S|}}{1-\norm{\Influ{\theta_S}}}$, by Lemma~\ref{lem:singular values under Dobrushin}, we have $\lsv{B}\leq 2\sqrt{\sum_{i\in S}v_i^2}$ and $\ssv{B}\geq \sqrt{\sum_{i\in S}v_i^2}/4$ with probability at least $1-2\exp\left(-\frac{\left(1-\norm{\Influ{\theta_S}}\right)\cdot (\sum_{i\in S}v_i^2)^2}{64c^4 k|S|}\right)\geq 1-2\exp(-4k)$. Therefore, 
$\normI{\hat{z}-z(t)}\leq \frac{32\sqrt{|S|k}}{\sqrt{\sum_{i\in S}v_i^2}}\cdot (\eta+\varepsilon)$ with probability at least $1-2\exp\left(-\frac{\left(1-\norm{\Influ{\theta_S}}\right)\cdot (\sum_{i\in S}v_i^2)^2}{64c^4 k|S|}\right)$.

\paragraph{Query Complexity in Different Models:} We set $\eta$ to be $\varepsilon$. 
\begin{itemize}
    \item \textbf{Deterministic structure:} we have a $\left(\varepsilon,\frac{4 \cdot \max_{j\in[k]} C_{jj} }{\alpha\beta}\cdot \varepsilon\right)$-query protocol using $k(\log{\normI{A}}+2\log (1/\varepsilon))$ queries.
    \item \textbf{Multivariate Gaussian distributions:} with probability at least $1-2\exp\left(-\frac{\Tr{\Sigma_S}}{16\cdot \rho(\Sigma_S)}\right)$ (no less than $1-2\exp(-4k)$ by our choice of $S$), we have a $\left(\varepsilon,\frac{64\sqrt{|S| k}}{\sqrt{\Tr{\Sigma_S}}} \cdot \varepsilon\right)$-query protocol using $|S|(\log{\normI{A}}+2\log (1/\varepsilon))$ queries.
    \item \textbf{Weakly dependent distributions:} with probability at least $1-2\exp\left(-\frac{\left(1-\norm{\Influ{\theta_S}}\right)\cdot (\sum_{i\in S}v_i^2)^2}{64c^4 k|S|}\right)$ (no less than $1-2\exp(-4k)$ by our choice of $S$), we have a $\left(\varepsilon,\frac{64\sqrt{|S| k}}{\sqrt{\sum_{i\in S}v_i^2}} \cdot \varepsilon\right)$-query protocol using $|S|(\log{\normI{A}}+2\log (1/\varepsilon))$ queries.
\end{itemize}
\end{prevproof}


\section{Bounding the Largest and Smallest Singular Values }
We prove both Lemma~\ref{lem:GGM singular values} and~\ref{lem:singular values under Dobrushin} using an $\varepsilon$-net argument. We first state a lemma that says that for any matrix $M$, if we can bound the maximum value of $\norm{Mx}$ over all points $x$ in the $\varepsilon$-net, then we also bound the largest and smallest singular values of $M$.
\begin{lemma}[Adapted from~\cite{rudelson2014recent}]\label{lem:eps-net}
For any $\varepsilon<1$, there exists an $\varepsilon$-net $\mathcal{K}\subseteq S^{n-1}$, i.e.,  $\forall x\in S^{n-1}~\exists y\in \mathcal{K}~ \norm{x-y}<\varepsilon$, such that $|\mathcal{K}|\leq (3/\varepsilon)^n$. For any matrix $M\in \mathbb{R}^{m\times n}$, let $a = max_{x\in \mathcal{K}} \norm{Mx}$ and $b=min_{x\in \mathcal{K}} \norm{Mx}$, then $\lsv{M}\leq \frac{a}{1-\varepsilon}$ and $\ssv{M}\geq b-\frac{\varepsilon}{1-\varepsilon}\cdot a$.
\end{lemma}
\begin{prevproof}{Lemma}{lem:eps-net}
Let $x^*\in S^{n-1}$ be a vector that satisfies $\norm{Mx^*}=\lsv{M}$. Let $x$ be a vector in $\mathcal{K}$ such that $\norm{x-x^*}\leq \varepsilon$. Then $\lsv{M}=\norm{Mx^*}\leq \norm{Mx}+\norm{M(x-x^*)}\leq a+\varepsilon \lsv{M}$, which implies that $\lsv{M}\leq \frac{a}{1-\varepsilon}$. On the other hand, for any $y\in S^{n-1}$, let $y'\in \mathcal{K}$ satisfies $\norm{y-y'}\leq \varepsilon$, then $\norm{My}\geq \norm{My'}-\norm{M(y-y')}\geq b- \varepsilon\cdot \lsv{M}\geq b- \frac{\varepsilon}{1-\varepsilon}\cdot a$. 
\end{prevproof}

\subsection{Multivariate Gaussian Distributions}\label{sec:proof of concentration Gaussian}
In this section, we prove the case where the columns of the random matrix are drawn from a  multivariate Gaussian distribution.  The key is again to prove that for every unit-vector, $\norm{Ux}$ lies between $[c_1\cdot\E[\norm{Ux}], c_2\cdot\E[\norm{Ux}]]$ with high probability for some absolute constant $c_1$ and $c_2$ (Lemma~\ref{lem:GGM concentration for every unit-vector}). Lemma~\ref{lem:GGM singular values} follows from the combination of Lemma~\ref{lem:GGM concentration for every unit-vector},~\ref{lem:eps-net}, and the union bound.

\begin{prevproof}{Lemma}{lem:GGM singular values}
 Let $Y^{(1)},\ldots, Y^{(n)}$ be $n$ i.i.d. samples from the distribution $\normal(0,I_m)$, and $V:=D^{1/2} [Y^{(1)},\ldots, Y^{(s)}]$.

\begin{proposition}\label{prop:change of variable}
$\normal(0,\Sigma)\disteq \Lambda^T\circ \normal(0,D)$ and $U\disteq \Lambda^T V$.
\end{proposition}
\begin{proof}
 $\E[\Lambda^T D^{1/2} Y^{(i)} (Y^{(i)})^T D^{1/2}\Lambda] = \Lambda^T D^{1/2}\E[Y^{(i)} (Y^{(i)})^T ]D^{1/2} \Lambda= \Lambda^T D\Lambda= \Sigma$.
\end{proof}

Since $\Lambda$ is an orthonormal matrix, $\lsv{U}=\lsv{V}$ and $\ssv{U}=\ssv{V}$. We will proceed to show that both $\lsv{V}$ and $\lsv{V}$ concentrate around their means. We do so via an $\varepsilon$-net argument.

\begin{lemma}\label{lem:GGM concentration for every unit-vector}
         For any fix $x\in S^{n-1}$, $\E[\norm{Vx}^2]=\Tr{D}$. Moreover, 
         $$\Pr\left[\norm{Vx}^2 \leq \frac{\Tr{D}}{4} \right]\leq \exp\left(-\frac{\Tr{D}}{8\cdot d_{max}}\right),$$ and 
                $$\Pr\left[\norm{Vx}^2\geq 2 \Tr{D} \right]\leq \exp\left(-\frac{\Tr{D}}{4\cdot d_{max}}\right).$$
\end{lemma}
\begin{prevproof}{Lemma}{lem:GGM concentration for every unit-vector}
Let $g_1,\ldots, g_n$ to be $n$ i.i.d. samples from $\normal(0,1)$. It is not hard to see that $Vx\disteq (\sqrt{d_1}g_1,\ldots,\sqrt{d_n}g_n)^T$, so we need to prove that $\sum_{i\in[n]} d_i g_i^2$ concentrates around its mean $\Tr{D}$.

\begin{align*}
    &\Pr\left[\sum_{i\in[n]} d_i g_i^2\leq \Tr{D}-t\right]\\
    =& \Pr\left[\exp\left(\lambda\cdot (\Tr{D}-\sum_{i\in[n]} d_i g_i^2)\right)\geq \exp(\lambda t)\right]~\quad(\text{$\lambda>0$ and will be specified later})\\
    \leq & \frac{\exp(\lambda \Tr{D}
    )\E\left[ \exp\left(-\lambda\cdot \sum_{i\in[n]} d_i g_i^2\right)\right]}{\exp(\lambda t)}= \frac{\exp(\lambda \Tr{D}
    )\prod_{i\in[n]}\E\left[ \exp\left(-\lambda\cdot d_i g_i^2\right)\right]}{\exp(\lambda t)}
\end{align*}

Since $g_i^2$ distributes according to a chi-square distribution, its moment generating function $$\E\left[ \exp\left(-\lambda\cdot d_i g_i^2\right)\right]=\frac{1}{\sqrt{1+2\lambda d_i}}.$$ If we choose $\lambda$ to be no more than $1/2d_{max}$, since for any $a\in[0,1]$, $1+2a\geq e^a$, we have that $$\frac{1}{\sqrt{1+2\lambda d_i}}\leq \exp(-\lambda d_i/2).$$

Putting everything together, we have that $$\Pr\left[\sum_{i\in[n]} d_i g_i^2\leq \Tr{D}-t\right]\leq \exp\left (-\lambda\cdot (t-\Tr{D}/2)\right).$$ When we choose $\lambda=1/2d_{max}$ and $t=3/4\cdot \Tr{D}$, the RHS of the inequality becomes $\exp\left(-\frac{\Tr{D}}{8\cdot d_{max}}\right)$.

Next, we upper bound $\Pr\left[\sum_{i\in[n]} d_i g_i^2\geq \Tr{D}+t\right]$ via a similar approach.
\begin{align*}
    &\Pr\left[\sum_{i\in[n]} d_i g_i^2\geq \Tr{D}+t\right]\\
    =& \Pr\left[\exp\left(\lambda\cdot (\sum_{i\in[n]} d_i g_i^2-\Tr{D})\right)\geq \exp(\lambda t)\right]~\quad(\text{$\lambda>0$ and will be specified later})\\
   \leq & \frac{\prod_{i\in[n]}\E\left[ \exp\left(\lambda\cdot (d_i g_i^2-d_i)\right)\right]}{\exp(\lambda t)}
\end{align*}
Note that $\E\left[ \exp\left(\lambda\cdot (d_i g_i^2-d_i)\right)\right]=\frac{\exp(-\lambda d_i)}{\sqrt{1-2\lambda d_i}}$.
\begin{proposition}\label{prop:approx}
For any $x\in [0, 1/4]$, $\frac{\exp(-x)}{\sqrt{1-2x}}\leq \sqrt{1+2x}$.
\end{proposition}
\begin{prevproof}{Proposition}{prop:approx}
We first state a few inequalities that are not hard to verify. First, for all $x>0$, $e^{-x}\leq 1-x+x^2$. Second, $\sqrt{1-4x^2}\geq 1-2x^2-8x^4$ if $x\in [0,1/2)$. Finally, $1-2x^2-8x^4\geq 1-x+x^2$ if $x\in[0,1/4]$. Combining all three inequalities, we have that $$e^{-x}\leq \sqrt{1-4x^2}=\sqrt{1-2x}\sqrt{1+2x},~\text{for all $x\in[0,1/4]$}.$$
\end{prevproof}

If we choose $\lambda$ to be no more than $1/4d_{max}$, then by Proposition~\ref{prop:approx}, $\frac{\exp(-\lambda d_i)}{\sqrt{1-2\lambda d_i}}\leq \sqrt{1+2\lambda d_i}$, which is upper bounded by $\exp(\lambda d_i)$. Putting everything together, we have that 
$$\Pr\left[\sum_{i\in[n]} d_i g_i^2\geq \Tr{D}+t\right]\leq \exp\left(-\lambda(t-\Tr{D})\right).$$
When we choose $\lambda=1/4d_{max}$ and $t=2 \Tr{D}$, the RHS of the inequality becomes $\exp\left(-\frac{\Tr{D}}{4\cdot d_{max}}\right)$.
\end{prevproof}

Next, we only consider when the good event happens, that is,  for all points $x$ in the $\varepsilon$-net, $\norm{Vx}\in \left[\frac{\sqrt{\Tr{D}}}{2},\sqrt{2\Tr{D}}\right]$. Combining Lemma~\ref{lem:GGM concentration for every unit-vector} and the union bound, we know that the good event happens with probability at least $1-2\exp\left(-\frac{\Tr{D}}{8\cdot d_{max}}+\ln(3/\varepsilon)\cdot n\right)$. According to Lemma~\ref{lem:eps-net}, $\lsv{V}\leq \frac{\sqrt{2\Tr{D}}}{1-\varepsilon}$ and $\ssv{V}\geq \frac{\sqrt{\Tr{D}}}{2}-\frac{\varepsilon}{1-\varepsilon}\cdot \sqrt{2\Tr{D}}$. If we choose $\varepsilon=1/7$, then $\lsv{V}\leq 2\sqrt{\Tr{D}}$ and $\ssv{V}\geq \frac{\sqrt{\Tr{D}}}{4}$.
\end{prevproof}

\subsection{Bounded Distributions with Weak Dependence}\label{sec:proof of concentration Dobrushin}
In this section, we prove the case where the columns of the random matrix are drawn from a $m$-dimensional distribution that satisfies weak dependence. The overall plan is similar to the one for multivariate Gaussian distributions. The key is again to prove that for every unit-vector, $\norm{Ux}$ lies between $\left[c_1\cdot\E[\norm{Ux}], c_2\cdot\E[\norm{Ux}]\right]$ with high probability for some absolute constant $c_1$ and $c_2$ (Lemma~\ref{lem:Dobrushin concentration for every unit-vector}). Lemma~\ref{lem:singular values under Dobrushin} then follows from the combination of Lemma~\ref{lem:Dobrushin concentration for every unit-vector},~\ref{lem:eps-net}, and the union bound.

\begin{prevproof}{Lemma}{lem:singular values under Dobrushin}

We first show that for each fix $x\in S^{n-1}$, $\norm{Ux}$ is concentrates around its mean. Then, we apply Lemma~\ref{lem:eps-net} to bound $\lsv{U}$ and $\ssv{U}$.
\begin{lemma}\label{lem:Dobrushin concentration for every unit-vector}
        Let $U=[X^{(1)},\ldots, X^{(n)}]$ be a $m\times n$ random matrix, where each column of $U$ is an independent copy of a $m$-dimensional random vector $X$. We assume that the coordinates of $X$ are weakly dependent, i.e., $\norm{\Influ{X}}<1$, and each coordinate of $X$ lies in $[-c,c]$ and has mean $0$ and variance $v_i^2$. Let $v=\sqrt{\sum_{i\in[m]} v_i^2}$. For any fix $x\in S^{n-1}$, $\E[\norm{Ux}^2]= v^2$ and 
        $$\Pr\left[|\norm{Ux}^2-v^2|>t\right]\leq 2 \exp\left(-\frac{\left(1-\norm{\Influ{X}}\right)t^2}{16c^4 nm}\right)$$
\end{lemma}

\begin{prevproof}{Lemma}{lem:Dobrushin concentration for every unit-vector}
We first expand $\norm{Ux}^2$. \begin{align*}
    \norm{Ux}^2
    =\sum_{i\in[m]}\left(\sum_{j\in[n]} u_{ij}x_j\right)^2
    =\sum_{i\in[m]}\left(\sum_{j\in[n]} u_{ij}^2x_j^2+2\sum_{k\neq j} u_{ij}u_{ik}x_jx_k\right).
\end{align*}
Therefore, $\E\left[\norm{Ux}^2\right]=\sum_{i\in[m]} v_i^2=v^2$. To prove that $\norm{Ux}^2$ concentrates, we first need a result by Chatterjee~\cite{Chatterjee2005a}.

\begin{lemma}[Adapted from Theorem 4.3 in~\cite{Chatterjee2005a}]\label{lem:concentration weak dependence}
Let $X$ be a $d$-dimensional random vector. Suppose function $f$ satisfies the following generalized Lipschitz condition: $$|f(x)-f(y)|\leq \sum_{i\in[d]} c_i \mathds{1}[x_i\neq y_i],$$ for any $x$ and $y$ in the support of $X$. If $\Influ{X}<1$, we have $$\Pr\left[ |f(X)-\E[f(X)]|\geq t\right]\leq 2\exp\left(-\frac{\left(1-\norm{\Influ{X}}\right)t^2}{\sum_{i\in[d]}c_i^2}\right).$$
\end{lemma}

The function we care about is $\norm{Ux}^2$, where the variables are $\{u_{ij}\}_{i\in[m],j\in[n]}$. If $U$ and $U'$ only differs at the $(i,j)$ entry, then 
\begin{align*}
    &|\norm{Ux}^2-\norm{U'x}^2|\\
    =&| u_{ij}^2 x_j^2+2\sum_{k\neq j} u_{ij}u_{ik}x_j x_k-(u'_{ij})^{2} x_j^2-2\sum_{k\neq j} u'_{ij}u_{ik}x_jx_k|\\
    \leq & c^2 x_j^2+4c^2|x_j||x_k|\leq 4c^2 |x_j|\left(\sum_{k\in[n]}|x_k|\right)\leq  4c^2\sqrt{n} |x_j|
\end{align*}
We denote $4c^2\sqrt{n} |x_j|$ by $c_{ij}$. Clearly, for any $U$ and $U'$, $|\norm{Ux}^2-\norm{U'x}^2|\leq \sum_{i,j\in[d]} c_{ij}\mathds{1}[u_{ij}\neq u'_{ij}]$. Also, notice that $\Influ{U}=I_n\otimes \Influ{X}$, and therefore $\norm{\Influ{U}}=\norm{\Influ{X}}$.~\footnote{$\otimes$ denotes the Kronecker product of the two matrices.} We apply Lemma~\ref{lem:concentration weak dependence} to $\norm{Ux}^2$ and derive the following inequality:
$$\Pr\left[|\norm{Ux}^2-v^2|>t\right]\leq 2 \exp\left(-\frac{\left(1-\norm{\Influ{X}}\right)t^2}{\sum_{i\in[m],j\in[n]}c_{ij}^2}\right)=2 \exp\left(-\frac{\left(1-\norm{\Influ{X}}\right)t^2}{16c^4 nm}\right).$$
\end{prevproof}
Next, we only consider when the good event happens, that is,  for all points $x$ in the $\varepsilon$-net, $\norm{Ux}\in \left[\frac{v}{2},\sqrt{2}v\right]$. Combining Lemma~\ref{lem:Dobrushin concentration for every unit-vector} (setting $t=3/4 v^2$) and the union bound, we know that the good event happens with probability at least $1-2\exp\left(-\frac{\left(1-\norm{\Influ{X}}\right)9v^4}{256c^4 nm}+\ln(3/\varepsilon)\cdot n\right)$. According to Lemma~\ref{lem:eps-net}, $\lsv{U}\leq \frac{\sqrt{2}v}{1-\varepsilon}$ and $\ssv{U}\geq \frac{v}{2}-\frac{\varepsilon}{1-\varepsilon}\cdot \sqrt{2}v$. If we choose $\varepsilon=1/7$, then $\lsv{U}\leq 2v$ and $\ssv{U}\geq \frac{v}{4}$.
\end{prevproof}

\bibliographystyle{plain}
\bibliography{Yang.bib}

\appendix
\section{Missing Proof of Lemma~\ref{lem:robustness}}\label{sec:robustness proof}
\begin{prevproof}{Lemma}{lem:robustness}
The proof essentially follows from the same analysis as Theorem 3 in~\cite{brustle2020multi}. We only provide a sketch here. Since we are working with the matrix factorization model and can directly exploit the low dimensionality of the latent representation, we manage to replace the dependence on $N$ with $\normI{A}$ in both the revenue loss and violation of the truthfulness constraints. 
Our proof relies on the idea of \emph{``simultaneously coupling''} by Brustle et al.~\cite{brustle2020multi}. More specifically, it couples $\widehat{F}_{z,i}$ with every distribution $F_{z,i}$ in the $\varepsilon$-Prokhorov-ball around $\widehat{F}_{z,i}$. If we round both $\widehat{F}_{z,i}$ and any $F_{z,i}$ to a random grid $G$ with size $\delta$, we can argue that the \emph{expected total variation distance} (over the randomness of the grid) between the two rounded distributions is $O(\varepsilon+\frac{\varepsilon}{\delta})$ (using Theorem 2 in~\cite{brustle2020multi}). Now consider the following mechanism: choose a random grid $G$, round the bids to the random grid, apply the  mechanism $M_G$ that we designed for the rounded distribution of $\bigtimes_i \widehat{F}_{z,i}$. More specifically, $M_G$ is the following mechanism: for each bid $b$, use $\condsample_i(b_i,\delta)$ to sample a bid $b'_i$ and run $\givenMech$ on the bid profile $(b'_1,\ldots, b'_m)$. 
 Since the expected total variation distance (over the randomness of the grid) between the two rounded distributions is $O(\varepsilon+\frac{\varepsilon}{\delta})$, we only need to argue that when the given distribution and the true distribution are close in total variation distance, we can robustify the mechanism designed for one distribution for the other distribution. This is a much easier task, and we again use a similar argument in~\cite{brustle2020multi} to prove it. Combining everything,  we can show that the randomized mechanism we constructed is  approximately-truthful and only loses a negligible revenue compared to $\givenMech$ under any distribution that is within the $\varepsilon$-Prokhorov-ball around the given distribution.	
\end{prevproof}

\end{document}